\documentclass{ifacconf}

\usepackage{graphicx}      
\usepackage[round]{natbib}
\usepackage{amssymb}
\usepackage{amsmath}
\usepackage{bbm}

\usepackage{enumitem}
\setenumerate{nolistsep}
\setitemize{nolistsep}
\usepackage{dsfont}
\usepackage{tikz}
\usetikzlibrary{automata,positioning}
\usepackage{mathtools}
\usepackage{todonotes}
\allowdisplaybreaks
\usepackage{algorithm,algorithmic}

\newtheorem{thrm}{\textbf{Theorem}} 
\newtheorem{defin}{\textbf{Definition}}
\newtheorem{facts}{\textbf{Fact}}

\newtheorem{props}{\textbf{Proposition}}
\newtheorem{remark}{\textbf{Remark}} 

\newtheorem{probs}{\textbf{Problem}}

\newtheorem{coro}{\textbf{Corollary}}

\newtheorem{lemma}{\textbf{Lemma}}

\renewcommand{\geq}{\geqslant}

\renewcommand{\leq}{\leqslant}

\newcommand{\tendsto}{\rightarrow}

\newcommand{\abs}[1]{\left\lvert{#1}\right\rvert}

\newcommand{\norm}[1]{\left\lVert#1\right\rVert}
\newcommand{\pmat}[1]{\begin{pmatrix}#1\end{pmatrix}}
\newenvironment{proof}{\paragraph{Proof:}}{\hfill$\square$}

\newcommand{\R}{\mathbb{R}}
\newcommand{\N}{\mathbb{N}}

\newcommand{\is}{i_{s}}
\newcommand{\iu}{i_{u}}

\newcommand{\Nsets}{\mathcal{Q}}

\newcommand{\gdash}{\mathcal{G}^{\prime}}
\newcommand{\vbar}{\overline{v}}

\newcommand{\Vprime}{\mathcal{V}^{\prime}}

\newcommand{\Edash}{\mathcal{E}^{\prime}}
\newcommand{\G}{\mathcal{G}(\mathcal{V}, \mathcal{E})}
\newcommand{\V}{\mathcal{V}}
\newcommand{\E}{\mathcal{E}}
\newcommand{\walk}{v_0, (v_0, v_1), v_1, \ldots, v_{n-1}, (v_{n-1}, v_0), v_0}

\newcommand{\cycle}{\overline{v}_0, (\overline{v}_0, \overline{v}_1), \overline{v}_1, \ldots, \overline{v}_{n-1}, (\overline{v}_{n-1}, \overline{v}_0), \overline{v}_0}
\newcommand{\wbar}{\overline{w}}
\newcommand{\wunder}{\underline{w}}

\DeclareMathOperator{\minimize}{minimize}
\DeclareMathOperator{\sbjto}{subject\;to}      
\begin{document}
\begin{frontmatter}

\title{Stabilizing scheduling logic for networked control systems under limited capacity and lossy communication networks} 


\author[First]{Anubhab Dasgupta} 

\address[First]{Department of Mechanical Engineering, Indian Institute of Technology Kharagpur, West Bengal 721302, India (email : anubhab.dasgupta@kgpian.iitkgp.ac.in).}

\begin{abstract}                
In this paper we address the problem of designing scheduling logic for stabilizing Networked Control Systems (NCSs) with plants and controllers remotely-located over a limited capacity communication network subject to data losses. Our specific contributions include characterization of stability under worst case data loss using an inequality associated with a cycle on a graph. This is eventually formulated as a feasibility problem to solve for certain parameters (\(T\)-factors) used to design a periodic scheduling logic. We show that given a solution to the feasibility problem, the designed scheduling logic guarantees \emph{global asymptotic stability} for all plants of the network under all admissible data losses. We also derive sufficient conditions on the number of plants and the capacity of the network for the existence of a solution to the feasibility problem. Given that a sufficient condition is satisfied, we discuss the procedure to obtain the feasible \(T\)-factors. We use tools from switched systems theory and graph theory in this work. A numerical experiment is provided to verify our results.
\end{abstract}


\end{frontmatter}

\section{Introduction}
\label{s:intro}
Networked Control Systems (NCSs) are spatially distributed systems in which the communication between plants and their controllers takes place over a shared communication network. They are omnipresent in modern-day Cyber-Physical System applications including that in a platoon of vehicles, in an industrial assembly line, and even in medical applications like remote surgery. In most of these applications, the capacity of the communication network is constrained. A scenario where the number of plants sharing a network is higher than the capacity of the network is known as a \emph{medium access constraint}. The problem of allocating the shared network to the plants at each time instant, is referred to as a \emph{scheduling problem} and the corresponding logic obtained is known as a \emph{scheduling logic}. These logics generally fall into two main categories: static (periodic) and dynamic (closed-loop) scheduling. In the case of static scheduling, an offline, finite-length allocation scheme for the shared communication channel is applied in a periodic manner. In this paper we design a static scheduling logic. See \citep{ref:zhang-survey-2016} and the references therein for a survey of recent results in NCS.

 The problem of designing scheduling policies has been widely studied, see for example \citep{Hristu-Varsakelis2005}, \citep{7396942}, \citep{7479071} and the references therein, where the authors have used tools such as common lyapunov functions, multiple lyapunov functions and linear matrix inequalities. Additionally, communication networks may experience uncertainties, such as data loss, which is typical in noisy networks. The study of how NCSs are affected under uncertain communication has recently received considerable attention: both in the case of attacks \citep{heemels-dos-attacks} and in the case of data losses \citep{KUNDU2022151}. In the recent work \citep{icc-23}, the authors consider Bernoulli packet-drops and provide BMI conditions to design a probabilistic scheduling logic for preserving the GAS of an NCS.  \\  
 We consider \(N\) discrete-time linear plants with controllers connected over a shared network that has a capacity of \(M(<N)\). We consider a data loss model wherein at a time instant some or all of the active plants in the network can be affected resulting in open-loop operation. However, the data losses are such that they can occur consecutively for at most \(\ell\) time instants, where \(\ell\) is pre-specified. A data loss acknowledgement signal or channel feedback affects the optimal control or estimation strategy. We solve our problem under the assumption that there is no such acknowledgement signal available. In \citep{dey-schenato}, the authors solve a state estimation problem for a single plant affected by quantization noise and an erasure channel. They show that the optimal strategy is to forward an innovation term through the channel when acknowledgement is available and just the state measurement otherwise. In our setting, we design a scheduling logic which guarantees the \emph{GAS} of all the plants in the network under all admissible data loss patterns. The challenge in the absence of data loss acknowledgement is that using a particular scheduling logic cannot guarantee the exact duration for which the plants actually operated in closed loop. \citep{atreyee-tcns20} has utilised the idea of identifying \emph{\(T\)-contractive cycles} for studying the design of stabilizing scheduling logic for NCSs assuming the communication to be ideal. Motivated by similar tools, we use a switched systems representation of our NCS and represent the operation of activating plants as traversing over a graph. To guarantee stability under data losses, we analyse the worst-case scenario with maximal admissible data loss and characterize GAS in such a case using \emph{contractive cycles} which are parameterized by \(T\)-factors. These \(T\)-factors are used to design the scheduling logic. Although a similar graph theoretic approach using \(T\)-contractive cycles was used in \citep{atreyee-tcns20}, the novelty in our work lies in the fact that we consider the problem in the presence of data losses. Our contribution in this paper is three-fold: (a) we characterize stability under the worst case using what we define as a  \emph{contractive cycle} on a graph. This cycle is parameterized by \(T\)-factors, the existence of which gives us such a \emph{contractive cycle}. We derive the characterization of this cycle, formulate it as a feasibility problem and provide an algorithm to design such a cycle and obtain the \(T\)-factors. (b) We propose an algorithm for designing a stabilizing scheduling logic without utilising a data loss acknowledgement signal. The scheduling logic uses the computed \(T\)-factors to compute the duration of activation for plants specified by the designed \emph{contractive cycle} and does so repeatedly. The scheduling logic being static with an offline computation, makes it suitable for safety-critical applications. (c) We first derive a general sufficient condition guaranteeing the existence of a \emph{contractive cycle} using a network partition based approach focussed towards dividing the feasibility problem into ones with smaller number of constraints. Followed by that we provide more specific conditions on \(N\) and \(M\) which guarantee solution to the feasibility problem and discuss their respective constructions. \\
 The remainder of the paper is organized as follows: in section \ref{s:prob_stat} we formally state our problem, section \ref{s:preliminaries} introduces the switched systems and graph theoretic representation of the problem which is followed by our main results, including algorithms and sufficient conditions, presented in section \ref{s:mainres}. We provide numerical experiment results in section \ref{s:numex} and conclude in section \ref{s:concln}. Detailed proofs of some of our results are provided in section \ref{s:proofs}. \\
{\it Notation}. We employ standard notation throughout the paper. \(\R\) is the set of real numbers and \(\N\) is the set of natural numbers. \(\N_0 \coloneqq \N \cup \{0\}\). \(\lambda_{\text{max}}(\mathcal{R})\) and \(\lambda_{\text{min}}(\mathcal{R})\) are used to denote the maximum and minimum eigenvalues of a matrix \(\mathcal{R} \in \R^{r\times r}\), respectively. \(a\%b\) denotes the remainder of the operation \(a/b\), for two scalars \(a\) and \(b\). \(\norm{v}\) denotes the Euclidean norm of a vector \(v\). \(F\) is called a \(b\)-vector with \(b \in \N\) if \(F \in \R^b\). For a finite set \(C\), its cardinality is denoted by \(\abs{C}\). \(\mathbbm{1}_{y}(y \in Y)\) is the indicator function which equals \(1\) when \(y \in Y\) and equals \(0\) when \(y \notin Y\). 
\section{Problem statement}
\label{s:prob_stat}
    We consider an NCS with \(N\) discrete-time linear plants whose dynamics are given by:
    \begin{align}
        \label{e:plants}
        x_i(t+1) = A_i x_i(t) + B_i u_i(t), \quad x_i(0) = x_i^0, \quad t \in \N_0,
    \end{align}
    where \(x_i(t) \in \R^{d_i}\) and \(u_i(t) \in \R^{m_i}\) are the vectors of states and inputs of the \(i\)-th plant at time \(t\) respectively. Each plant \(i\) has access to a remotely located state-feedback controller given by \(u_i(t) = K_i x_i(t)\), \(i=1,2,\ldots,N\). The matrices \(A_i \in \R^{d_i \times d_i}\), \(B_i \in \R^{d_i \times m_i}\) and \(K_i \in \R^{m_i \times d_i}\), \(i = 1, 2, \ldots, N\) are constant. Thus, each of the \(N\) plant-controller pairs communicate over a forward and reverse channel. We assume that the open-loop dynamics of each plant is unstable and each controller is stabilizing. More specifically, the matrices \(A_i+B_iK_i\), \(i=1,2,\ldots,N\) are Schur stable and the matrices \(A_i\), \(i = 1, 2, \ldots, N\) are unstable, that is, they are not Schur stable.\footnote{A matrix \(A \in \R^{d_i \times d_i}\) is Schur stable if all its eigenvalues lie within the open unit disk and it is non-Schur otherwise}
    
    The plants communicate with their respective controllers over a shared communication network with the following properties:
    \begin{itemize}[leftmargin=*]
        \item it has a limited communication capacity in the sense that at any time instant only \(M\) plants \((0<M<N)\) can access the network. Consequently the remaining \(N-M\) plants operate in open-loop, that is, with \(u_i(t)=0\),
        \item all the forward channels (from the plant to the controller) are ideal with perfect communication. The reverse channels (controller to plant) are degraded in the sense that they are prone to data losses wherein in the event of a data loss affecting the communication network at a time instant \(t\), the control signals are lost while transmission across some or all of the active communication channels in the network. More specifically, if data loss occurs in a channel \(j \in \{1,2,\ldots,M\}\) at time \(t\), the control input of plant \(i \in \{1,2,\ldots,N\}\) accessing the channel \(j\) at time \(t\) is lost, then the plant \(i\) operates in open-loop, that is, with \(u_i(t)=0\). Moreover, the number of consecutive data loss instants is bounded above by a pre-specified \(\ell \in \N\). 
    \end{itemize}
    Let \(\kappa_j \colon \N_0 \rightarrow \{0,1\}\) denote the data loss signal at channel \(j \in \{1,2,\ldots,M\}\). If \(\kappa_j(t)=0\), then the plant \(i \in \{1,2,\ldots,N\}\) accessing the channel \(j\) at time \(t\) receives its control input, and if \(\kappa_j(t)=1\), then the control input for the plant \(i \in \{1,2,\ldots,N\}\) accessing the channel \(j\) at time \(t\) is lost in the network. Let \(\mathcal{D}(\ell)\) denote the set of all \(\kappa_j\), \(j=1,2,\ldots,M\) for which the number of consecutive data losses is at most \(\ell\). A data loss signal \(\kappa\) is called \textit{admissible} if it administers data losses such that they can occur consecutively for at most \(\ell\) time instants.
    Notice that under any admissible data loss signal, the control inputs sent through \emph{some or all} of the \(M\) channels being accessed at a time \(t \in \N_0\), may be lost. \\
    We define \(\mathcal{S}\) to be the set of all \(M\)-vectors, from \(\{1,2,\ldots,N\}\) with distinct elements. We call \(\gamma : \N_0 \longrightarrow \mathcal{S}\) to be the \emph{scheduling logic} that determines which \(M\) plants will get access to the network, at each time \(t\), \(t \in \N_0\). The remaining \(N - M\) plants operate in open loop. There exists a diverging sequence of times \(0 \eqqcolon \tau_0 < \tau_1 < \tau_2 < \ldots\) and a sequence of indices \(s_0, s_1, s_2, \ldots\), with \(s_h \in \mathcal{S}\), \(h = 0, 1, 2, \ldots\) such that \(\gamma(t) = s_h\) for \(t \in \left[\tau_h, \tau_{h+1}\right[\).  
    \begin{defin}\citep[Appendix A.1]{liberzonswitching03} 
        \label{defin:gas}
        Plant \(i\) in \eqref{e:plants} is said to be \emph{globally asymptotically stable} (GAS) for a given scheduling logic, \(\gamma\), if there exists a class \(\mathcal{KL}\) function \(\beta_i\) such that the following inequality holds:
        \[
            \norm{x_i(t)} \leq \beta_i(\norm{x_i(0)}, t) \quad \text{ for all \(x_i(0) \in \R^{d_i}\) and \(t \in \N_0\).\footnotemark}
        \]\footnotetext{Recall the classes of functions \(\mathcal{K} \coloneqq \{\psi \colon \left[0, + \infty\right[ \rightarrow \left[0, + \infty\right[ | \psi \text{ is continuous, strictly increasing and } \psi(0) = 0\}\), \(\mathcal{L} \coloneqq \{\phi \colon \left[0, + \infty\right[ \rightarrow \left[0, + \infty\right[ | \phi \text{ is continuous and } \phi(s) \searrow 0 \text{ as } s \nearrow + \infty \}\), and \(\mathcal{KL} \coloneqq \{\Lambda \colon \left[0, + \infty\right[^2 \rightarrow \left[0, + \infty\right[ | \Lambda(.,s) \in \mathcal{K} \text{ for each \(s\) and } \Lambda(r,.) \in \mathcal{L} \text{ for each }r\}\).}
    \end{defin}
    We study the conditions under which all the plants are globally asymptotically stable under any admissible data loss signal. We state our problem formally here.
    \begin{probs}
        \label{prob:main_prob}
        Given the plant dynamics, \((A_i,B_i)\), \(i=1,2,\ldots,N\), the controller dynamics, \(K_i\), \(i=1,2,\ldots,N\), the capacity of the communication network \(M(<N)\) and the number of maximum consecutive data losses, \(\ell\), design a scheduling logic \(\gamma\), under which each plant \(i=1,2,\ldots,N\) in \eqref{e:plants}, is \(\emph{globally asymptotically stable}\) under all admissible data loss signals \(\kappa_j \in \mathcal{D}(\ell)\), \(j=1,2,\ldots,M\).
    \end{probs}

    In the sequel we will refer to a scheduling logic \(\gamma\), that is obtained as a solution to Problem \ref{prob:main_prob} as a \emph{stabilizing scheduling logic}. We will provide a solution to Problem \ref{prob:main_prob} assuming that there is no data loss acknowledgement signal or channel feedback.

\section{Preliminaries}
\label{s:preliminaries}
    Similar to \citep{atreyee-tcns20} we model the plants of the NCS as switched systems and provide a graph theoretic representation of the same. 
    \subsection{\textbf{Switched system modelling}}
    The dynamics of the \(i\)-th plant, \(i=1,2,\ldots,N\) in \eqref{e:plants} can be expressed as a switched system \citep{liberzonswitching03}:
    \begin{align}
        \label{e:switched_system}
        x_i(t+1) = A_{\sigma_i(t)}x_i(t), \; x_i(0) = x_i^0, \; \sigma_i(t) \in \{\is, \iu\}, \; 
        t \in \N_0 
    \end{align}
    where the subsystems are \(\{A_{\is}, A_{\iu}\}\), and a switching logic \(\sigma_i:\N_0 \longrightarrow \{\is, \iu\}\), \(i=1,2,\ldots,N\) satisfies:
    \begin{align*}
        \sigma_i(t) = 
        \begin{cases}
            \is, & \text{if } \is \in \gamma(t) \text{ and } \kappa_i(t) \neq 1 \\
            \iu, & \text{if } \is \notin \gamma(t) \\
                 & \text{or when } \is \in \gamma(t) \text{ and } \kappa_i(t) = 1 
            \end{cases}
    \end{align*}
    The switching logic \(\sigma_i\), for each \(i = 1, 2, \ldots, N\), depends on the scheduling logic \(\gamma\) and the data loss signal \(\kappa\). Therefore, it suffices to demonstrate that, under a stabilizing scheduling logic and any allowable data loss signal \(\kappa \in \mathcal{D}(\ell)\), the logic \(\sigma_i\) guarantees global asymptotic stability (GAS) for all plants \(i = 1, 2, \ldots, N\). We summarize relevant properties of Lyapunov-like functions from recent studies for context.
    \begin{facts} \citep[Fact 1]{atreyee-hscc14}
        \label{facts:lyapunov-like-1}
        There exists pairs \((P_p, \lambda_p)\), \(P_p \in \R^{d_i \times d_i}\) where \(P_p\) are positive definite matrices and \(p \in \{\is, \iu\}\), with scalars \(0 < \lambda_{\is} < 1\) and \(\lambda_{\iu} > 1\), for every plant \(i=1,2,\ldots,N\) in \eqref{e:plants}, such that with:
        \begin{align}
            \label{e:lyapunov-like}
            \R^{d_i} \ni \zeta \mapsto V_p(\zeta) \coloneqq \langle P_p \zeta, \zeta \rangle \in \left[0, +\infty\right[            
        \end{align}
        we have the following inequality:
        \begin{align}
            \label{e:lyapunov-lambda}
            V_p(z_p(t+1)) \leq \lambda_p V_p(z_p(t)), \quad t \in \N_0
        \end{align}
        where \(z_p(\cdot)\) solves the \(p\)-th recursion in \eqref{e:switched_system}, \(p \in \{\is, \iu\}\).
    \end{facts}

    \begin{facts} \citep[Fact 2]{atreyee-hscc14}
        \label{facts:lyapunov-like-2}
        There exists scalars \(\mu_{pq} \geq 1\), \(p, q \in \{\is, \iu\}\), for every plant \(i=1,2,\ldots,N\) such that the following inequality holds:
        \begin{align}
            \label{e:lyapunov-mu}
            V_q(\zeta) \leq \mu_{pq}V_p(\zeta), \quad \zeta \in \R^{d_i}
        \end{align}
    \end{facts}
    The functions \(V_p\), \(p \in \{\is, \iu\}\), \(i=1,2,\ldots,N\) are called Lyapunov-like functions. From the definition of \(V_p\), they are linearly comparable. Lyapunov-like functions are widely used in the stability theory of switched and hybrid systems \citep{liberzonswitching03}, \citep{branicky-98}. The scalars \(\lambda_{\is}\), \(i=1,2,\ldots,N\) provide a measure of the contraction (stability) of the corresponding \(V_p\) when a plant is operating in one of the stable mode(s), while \(\lambda_{\iu}\) is associated with the expansion (instability) in the corresponding \(V_p\) when a plant is in the unstable mode of operation. A tight estimate of the scalars \(\mu_{pq}\), \(p,q \in \{\is, \iu\}\) was proposed as \(\lambda_{\text{max}}\left(P_q P_p^{-1}\right)\) \citep[Proposition 1]{atreyee-hscc14}. See \citep[Remark 6]{atreyee-tcns20} for a detailed discussion about estimating these scalars.

    \subsection{\textbf{Graph theoretic representation of the NCS}}
    A directed graph \(\G\) is defined with vertices connected by directed edges. The vertex set \(\V\) has \(\abs{\V} = {N \choose M}\) vertices, each uniquely labeled as \(L(v) = \{\ell_{v}(1), \ell_{v}(2), \ldots, \ell_{v}(N)\}\), where \(\ell_{v}(i) = \is\) for exactly \(M\) elements and \(\ell_{v}(i) = \iu\) for the remaining \(N - M\) elements. The edge set \(\E\) consists of directed edges \((u, v)\) between each distinct pair of vertices \(u, v \in \V\).
    The reader can refer to \citep{bollobas} for the definitions of a walk, closed walk and a cycle on a graph.\\
    Let \(c = \walk\) be a cycle on \(\G\). 
    Let the functions \(\wbar \colon \V \rightarrow \R^N\) and \(\wunder \colon \E \rightarrow \R^N\) associate weights to the vertices and edges of \(\G\) respectively. We define:
    
    {\small \begin{align*}
        \wbar(v) = \pmat{\wbar_1(v) \\ \wbar_2(v) \\ \vdots \\ \wbar_N(v)}, \quad v \in \V
    \end{align*}
    where,
    \begin{align}
        \label{e:vert_weights}
        \wbar_i(v) =
        \begin{cases}
             -\abs{\ln{\lambda_{\is}}}, & \text{ if } \ell_v(i) = \is, \\
             &\quad i=1,2,\ldots,N \\
              \quad\abs{\ln{\lambda_{\iu}}}, & \text{ if } \ell_v(i) = \iu 
        \end{cases}
    \end{align}
    and,
    \begin{align*}
        \wunder(u,v) = \pmat{\wunder_1(u,v) \\ \wunder_2(u,v) \\ \vdots \\ \wunder_N(u,v)}, \quad (u,v) \in \E 
    \end{align*}
    where,
    \begin{align}
        \label{e:edge_weights}
        \wunder_i(u,v) =
        \begin{cases}
             \ln{\mu_{\is \iu}}, & \text{ if } \ell_u(i) = \is, \ell_v(i) = \iu, \\
             \ln{\mu_{\iu \is}}, & \text{ if } \ell_u(i) = \iu, \ell_v(i) = \is, \\
             &\quad i=1,2,\ldots,N \\
             0, & \text{ otherwise }
        \end{cases}
    \end{align}
    }
    
    The label \(L(v)\) for each vertex \(v\) identifies the specific \(M\) plants that have access to the communication network (or are in closed-loop operation). When the scheduling logic activates vertex \(v\) at time \(t\), the \(i\)-th plant operates in a stable mode if \(\ell_v(i) = \is\) and in an unstable, open-loop mode if \(\ell_v(i) = \iu\). A directed edge from vertex \(u\) to \(v\) represents a transition in network access from the set of \(M\) plants designated by \(u\) to those designated by \(v\).
    \begin{defin} 
        \label{defin:contractive-sol-2}
        A cycle \(c=\walk\) on \(\G\) is called a \emph{contractive cycle} if there exists at least one set of positive integers \(T_{v_k}\), \(k=0,1,\ldots,n-1\), \(2\leq n \leq \abs{\V}\) such that the following inequality is satisfied for all \(i=1,2,\ldots,N\):
        \begin{align}
            \label{e:contractive-sol-2}
            \overline{Z}_i(c) &\coloneqq \left(\sum\limits_{k=0}^{n-1}\left(\mathbbm{1}_{\ell_{v_k}(i)}\left(\ell_{v_k}(i)=\is\right)\wbar_i(v_k) \right. \right. \notag \\ 
            &\quad \quad \qquad \left.+ \mathbbm{1}_{\ell_{v_k}(i)}\left(\ell_{v_k}(i)=\iu\right)\wbar_i(v_k)\times(\ell + 1) \right) \notag \\
            &\quad \qquad + \left.\left(\ln{\mu_{\ell_{v_k}(i)\iu}}+\ln{\mu_{\iu\ell_{v_k}(i)}}\right)\right)T_{v_k}<0.
        \end{align}
        where \(\wbar_i(v_k)\) is the \(i\)-th component of the vector \(\wbar(v_k)\). The positive integers \(T_{v_k}\), \(k=0,1,\ldots,n-1\), will be referred to as the \(T\)-factors associated with the cycle \(c\). 
    \end{defin}
    We move on to our main results.
\section{Main results}
\label{s:mainres}
    Our main result can be stated as the following theorem:
    \begin{thrm} 
        \label{thrm:main_res-sol-2}
        Consider the NCS with limited capacity under data losses as described in section \ref{s:prob_stat}. Let the matrices \(A_i\), \(B_i\), \(K_{i}\) be given for \(i=1,2,\ldots,N\) along with the constants \(M\) and \(\ell\). If there exists a cycle satisfying the condition in \eqref{e:contractive-sol-2}, then there exists a scheduling logic \((\gamma(t))_{t\geq 0}\) which preserves GAS of each plant in \eqref{e:plants} under all admissible data loss signals \(\kappa_j \in \mathcal{D}(\ell)\), \(j=1,2,\ldots,M\). Moreover, the stabilizing scheduling logic \(\gamma\), can be obtained using Algorithm \ref{algo:scheduling_policies-sol-2}. 
    \end{thrm}
    The proof of the above theorem is provided in section \ref{s:proofs}. \\
    In Algorithm \ref{algo:scheduling_policies-sol-2} we periodically activate the set of \(M\) plants specified by a vertex \(v_k\), \(k=0,1,\ldots,n-1\) for \(T_{v_k}(\ell + 1)\) consecutive time instants and the value \(\gamma\) at any instant is given by the set of active plants according to the vertex of the cycle. Algorithm \ref{algo:scheduling_policies-sol-2} is similar in principle to \citep[Algorithm 1]{atreyee-tcns20}. In our approach, rather than activating the plants associated with each vertex of the contractive cycle for exactly its \(T\)-factor duration, we extend the activation period to \(\ell + 1\) times the \(T\)-factor. The rationale for this adjustment will become evident in the lemmas presented next.

    The following technical lemmas reveal some useful extremal properties of the cycle which the scheduling logic \(\gamma\) uses. Further, these will be used in the proof of Theorem \ref{thrm:main_res-sol-2}.
    \begin{lemma}
        \label{lemma:t-sol-2}
        If a scheduling logic \(\gamma\) is such that it is static and activates a particular set of \(M\) channels or equivalently a specific vertex \(v\) for an interval of \(T_{v}(\ell + 1)\times \mathcal{K}\), where \(\mathcal{K}\) is a positive integer, then the minimum time in the entire interval for which the vertex \(v\) operates without any data losses is \(T_{v} \times \mathcal{K}\) and the maximum duration in the said interval for which it operates without data losses is \(T_{v}(\ell + 1)\times \mathcal{K}\).
    \end{lemma}
    \begin{proof}
        Given the nature of any admissible data loss where the data losses do not occur consecutively for more than \(\ell\) time instants, the statement of the above lemma follows from a simple usage of the \emph{Box Principle}.
    \end{proof}
    \begin{lemma}
        \label{lemma:total-time-sol-2}
        A scheduling logic \(\gamma\) obtained from Algorithm \ref{algo:scheduling_policies-sol-2} has the property that it takes exactly \(\sum\limits_{k=0}^{n-1}T_{v_k}(\ell + 1)\) consecutive time instants to complete one full cycle of activating every vertex \(v_k\), \(k=0,1,\ldots,n-1\) of the input cycle \(c\) and the duration of uninterrupted activation \(t_{v_k}\) of a vertex \(v_k\) in the above interval satisfying \(T_{v_k} \leq t_{v_k} \leq T_{v_k}(\ell + 1)\).
    \end{lemma}
    \begin{proof}
        The proof of the above lemma is immediate from the static nature of a scheduling logic designed using Algorithm \ref{algo:scheduling_policies-sol-2} and from Lemma \ref{lemma:t-sol-2} applied to each vertex consecutively in order.
    \end{proof}
    \subsection{On designing a contractive cycle}
    Theorem \ref{thrm:main_res-sol-2} provides an algorithm to design the scheduling logic given a contractive cycle as the input which ensures that all the plants follow globally asymptotically stable trajectories under any admissible data loss pattern. In this subsection we look closely into the problem of designing a contractive cycle and finding its associated \(T\)-factors. We solve this in two parts. First we discuss how to compute the scalars \(\lambda_{\is}\), \(\lambda_{\iu}\), \(\mu_{\is \iu}\) and \(\mu_{\iu \is}\) for all \(i=1,2,\ldots,N\) followed by a procedure to choose \emph{candidate} cycles on \(\G\). The second part involves finding the \(T\)-factors and we formulate it as a feasibility problem given a candidate cycle \(c\) on \(\G\) as the input with the associated scalars \(\lambda_{\is}\), \(\lambda_{\iu}\), \(\mu_{\is \iu}\) and \(\mu_{\iu \is}\) for all \(i=1,2,\ldots,N\). The procedure that we follow for the first and the second parts is independent of data losses affecting the network and thus closely follows \citep[Algorithm 2]{atreyee-tcns20}. Although we say that the second part follows a similar route as in \citep[Algorithm 2]{atreyee-tcns20}, all we mean is that the same tools can be used for \emph{solving} the feasibility problem. The actual expression for the constraint for our problem setting is particular to our problem setting of identifying \(T\)-factors which makes the plants \emph{GAS} under all admissible data loss patterns. We briefly describe the procedure here and the reader can refer to \citep[Algorithm 2]{atreyee-tcns20} and the discussions therein for further details. In order for condition \eqref{e:lyapunov-lambda} to be satisfied we need to solve for pairs \((P_p,\lambda_p)\), \(p \in \{\is,\iu\}\) such that 
    \begin{align}
    \label{e:algo-2}
        &A_{\is}P_{\is}A_{\is}^\top - \lambda_{\is}P_{\is} \preceq 0, \;\;\; P_{\is} \succeq 0, \; 0<&\lambda_{\is}<1, \notag \\
        &A_{\iu}P_{\iu}A_{\iu}^\top - \lambda_{\iu}P_{\iu} \preceq 0, \; P_{\iu} \succeq 0, \; &\lambda_{\iu} \geq 1
    \end{align}
    are satisfied simultaneously. The above are BMIs and in general being numerically hard to solve, we use a grid based approach to re-write them as LMIs and use standard LMI solvers. For a particular plant, we divide the interval \((0,1)\) uniformly into a sufficiently large number of partitions with their mid-points denoting possible choice of \(\lambda_{\is}\). Using \citep[Remark 6]{atreyee-tcns20} we once again divide the interval between \(0\) and \(1\) uniformly and now use the partition mid-points (let's call them \(\eta_j\), parameterized by a partition index \(j\)) to check if \(\eta_j A_{\iu}\) is Schur stable and if so we add \(\frac{1}{\eta_j^2}\) as one of the possible values for the scalar \(\lambda_{\iu}\). Now, we iterate over this two dimensional grid of all possible combinations of scalars \(\lambda_{\is}\) and \(\lambda_{\iu}\) to solve for \(P_p\), \(p \in \{\is,\iu\}\) in \eqref{e:algo-2} which now becomes a pair of LMIs due to choosing the values of the respective scalars. If there exists a solution to the simultaneous pair of LMIs the scalars \(\mu_{\is\iu}\) and \(\mu_{\iu\is}\) are chosen as \(\mu_{\is\iu} = \lambda_{\text{max}}\left(P_{\iu} P_{\is}^{-1}\right)\) and \(\mu_{\iu\is} = \lambda_{\text{max}}\left(P_{\is} P_{\iu}^{-1}\right)\) respectively (\citep[Proposition 1]{atreyee-hscc14}). This process is repeated for all plants \(i=1,2,\ldots,N\).
    For the second part, we choose a candidate cycle \(c\) on \(\G\), \(c = \walk\) and solve the following feasibility problem in \(T_{v_k}\), \(k=0,1,\ldots,n-1\) for all \(i=1,2,\ldots,N\):
    \begin{align}
        \label{e:feasibility-T}
        \minimize&\quad1 \notag \\
        \sbjto&\:
        \begin{cases}
            T_{v_k} > 0, \\
            \overline{Z}_i(c) < 0, 
        \end{cases}
    \end{align} 
    where \(\overline{Z}_i(c)\) is as defined in \eqref{e:contractive-sol-2}. \\
    \label{p:para}If there is a solution to the above feasibility problem, we proceed with using the cycle \(c\) to design \(\gamma\) using Algorithm \ref{algo:scheduling_policies-sol-2}.

    \begin{algorithm}[h]
	    \caption{Construction of scheduling logic} 
        \label{algo:scheduling_policies-sol-2}
		\begin{algorithmic}
			\renewcommand{\algorithmicrequire}{\textbf{Input:}}
			\renewcommand{\algorithmicensure}{\textbf{Output:}}
			
			\REQUIRE A \emph{contractive} cycle \(c = \walk\) on \(\G\) with associated \(T\)-factors: \(T_{{v}_k}\), \(k=0,1,\ldots,n-1\).
			\ENSURE A scheduling logic \(\gamma\).
			
			\STATE \textit{Step-I}: For each vertex \({v}_k\), \(k=0,1,\ldots,n-1\), pick the elements \(i\) with label \(\ell_{{v}_k}(i) = \is\), \(i=1,2,\ldots,N\), and construct \(M\)-dimensional vectors \(s_k\), \(k=0,1,\ldots,n-1\)
			    \FOR {\(k=0,1,\ldots,n-1\)}
                    \STATE set \(p = 0\)
                    \FOR {\(i=1,2,\ldots,N\)}
                        \IF {\(\ell_{{v}_k}(i) = \is\)}
                            \STATE set \(p=p+1\) and \(s_k(p) = \is\)
                        \ENDIF
                    \ENDFOR
                \ENDFOR
            \STATE \textit{Step-II}: Construct a scheduling logic \(\gamma\) using the vectors \(s_k\), \(k=0,1,\ldots,n-1\) obtained in \textit{Step-I} and \(T_{{v}_k}\), \(k=0,1,\ldots,n-1\)
                \STATE set \(p = 0\) and \(t = 0\)
                \FOR {\(r=pn,pn+1,\ldots,(p+1)n - 1\)}
                    \STATE set \(m=0\)
                    \WHILE {\(m < T_{{v}_r}(\ell + 1)\)}
                        \STATE set \(\gamma(t) = s_{r-pn}\), \(m=m+1\) and \(t=t+1\)
                    \ENDWHILE
                \ENDFOR
		          \STATE set \(p=p+1\) and go to 12. 
		\end{algorithmic}
	\end{algorithm}
 \subsection{Sufficient condition for the existence of a contractive cycle}
 It is of interest to know under what conditions we can guarantee the existence of a contractive cycle \(c\) on \(\G\). In this subsection we first provide a technical lemma followed by a sufficient condition for the existence of a contractive cycle which utilize subgraphs to exploit the combinatorial structure of the capacity constraint, which when satisfied, guarantee the existence of a contractive cycle.
    \begin{lemma}
        \label{lem:suff-cond-1}
        Consider the partitions \(\Nsets_j \subseteq\{1,2,\ldots,N\}\), \(j=1,2,\ldots,M\) and respective graphs \(\gdash_j(\Vprime_j, \Edash_j)\) constructed with \(\Nsets_j\) as the set of plants like in Section \ref{s:preliminaries}-B with communication capacity \(M_{\Nsets_j}=1\). If the following conditions hold:
            \begin{enumerate}[label=\alph*), leftmargin =*]
                \item \(\Nsets_m\cap\Nsets_n=\emptyset\) for all \(m,n=1,2,\ldots,M\), \(m\neq n\),
                \item \(\displaystyle{\bigcup_{j=1}^{M}\Nsets_j=\{1,2,\ldots,N\}}\), and
                \item graph \(\gdash_j(\Vprime_j, \Edash_j)\), admits contractive cycle \(c_j\) for all \(j=1,2,\ldots,M\). 
            \end{enumerate}
        Then there exists a cycle \(c\) on \(\G\) which solves the feasibility problem \eqref{e:feasibility-T}.
    \end{lemma}
    A concise proof of the above Lemma is provided in section \ref{s:proofs}.
    \begin{remark}
        \label{remark:prop-1}
        Note that the idea here is to exploit the graph-based representation of the network and divide the problem of identifying a contractive cycle on \(\G\) into independent problems of identifying the same on smaller graphs defined on the partitions. Lemma \ref{lem:suff-cond-1} highlights that if we can partition the set of \(N\) plants into \(M\) smaller, non-overlapping subsets \(\Nsets_j\), \(j=1, 2, \ldots, M\), and form \emph{smaller} graphs \(\gdash_j(\Vprime_j, \Edash_j)\) using \(\Nsets_j\) as the plant set with a communication capacity of \(1\) that support contractive cycles, then it is possible to construct a contractive cycle on the larger graph \(\G\), which uses the full set of \(N\) plants with a communication capacity of \(M\). 
    \end{remark}
    Given Lemma \ref{lem:suff-cond-1} holds, it will be shown in its proof that there is only one cycle for each \(\Nsets_j\), \(j=1,2,\ldots,M\) which can be the candidate for a contractive cycle. As a result, one only needs to check if inequality \eqref{e:contractive-sol-2} holds for that candidate cycle. Let the candidate cycle on \(\gdash_j(\Vprime_j, \Edash_j)\) associated with each partition \(\Nsets_j\) be \(c_j=v_1, (v_1,v_2), v_2, \ldots, v_{\abs{\Nsets_j}},(v_{\abs{\Nsets_j}},v_1),v_1\) with labels \(\ell_{v_k}(i)=\iu\) for all \(i=\Nsets_j\setminus \{k\}\) and \(\ell_{v_k}(i)=\is\) for \(i=k\), for all vertices \(v_k\), \(k=1,2,\ldots,\abs{\Nsets_j}\) with \(T\)-factors \(T^j_{v_k}\), \(k=1,2,\ldots,\abs{\Nsets_j}\). Then
    
    {\small {\begin{align}
        \label{e:cycle-cj}
        \overline{Z}_i(c_j) &= \left(\sum\limits_{k=1}^{\abs{\Nsets_j}}\left(\mathbbm{1}_{\ell_{v_k}(i)}\left(\ell_{v_k}(i)=\is\right)\wbar_i(v_k) \right. \right. \notag \\ 
        &\quad \qquad \left.+ \mathbbm{1}_{\ell_{v_k}(i)}\left(\ell_{v_k}(i)=\iu\right)\wbar_i(v_k)\times(\ell + 1) \right) \notag \\
        &\quad \qquad + \left.\left(\ln{\mu_{\ell_{v_k}(i)\iu}}+\ln{\mu_{\iu\ell_{v_k}(i)}}\right)\right){T}^j_{v_k} \notag \\
        &= \left(-\abs{\ln{\lambda_{i_s}}}+\left(\ln{\mu_{i_s i_u}}+\ln{\mu_{i_u i_s}}\right)\right)T^j_{v_i} \notag \\
        &\quad+ \left(\sum\limits_{k=1}^{\abs{\Nsets_j}}T^j_{v_k} - T^j_{v_i}\right)(\ell + 1)\ln{\lambda_{\iu}} < 0,
    \end{align}}}
    needs to hold for all plants \(i \in \Nsets_j\) for \(c_j\) to be a contractive cycle on \(\gdash_j(\Vprime_j, \Edash_j)\), \(j=1,2,\ldots,M\). 
    Next, we present a sufficient condition for \eqref{e:cycle-cj} to hold. 
    \begin{props}
        \label{props:suff-cond-2}
        Suppose that there exist \(\mathcal{P}_j \subseteq \{1,2,\ldots,N\}\), \(j=1,2,\ldots,M\) such that the following conditions hold:
            \begin{enumerate}[label=\alph*), leftmargin =*]
                \item \(\mathcal{P}_m\cap\mathcal{P}_n=\emptyset\) for all \(m,n=1,2,\ldots,M\), \(m\neq n\),
                \item \(\displaystyle{\bigcup_{j=1}^{M}\mathcal{P}_j=\{1,2,\ldots,N\}}\), and
                \item the following inequality holds for all \(i \in \mathcal{P}_j\), \(j=1,2,\ldots,M\):
            \end{enumerate}
            \begin{align}
                \label{e:suff-cond-2}
                \left(\abs{\ln{\lambda_{i_s}}}-\left(\ln{\mu_{i_s i_u}}+\ln{\mu_{i_u i_s}}\right)\right)>\left(\abs{\mathcal{P}_j}-1\right)(\ell + 1)\ln{\lambda_{\iu}}.
            \end{align}
        Then there exists a \emph{contractive cycle} on \(\G\).    
    \end{props}
    \begin{proof}
        The above conditions are a restatement of the conditions in Lemma \ref{lem:suff-cond-1} with an additional constraint imposed on the third condition of the previous Lemma. Indeed, substituting \(T^j_{v_k}=T^j \; \text{(const.)}\) for \(k=1,2,\ldots,\abs{\mathcal{P}_j}\), \(j=1,2,\ldots,M\) in \eqref{e:cycle-cj} leads to \eqref{e:suff-cond-2}. As a result \eqref{e:suff-cond-2} automatically implies that \eqref{e:cycle-cj} holds and thus is a sufficient condition for the same implying that we can use the procedure in the proof of Lemma \ref{lem:suff-cond-1} for constructing the corresponding contractive cycle on \(\G\). 
    \end{proof}
    \begin{props}
        If the following inequality 
        \begin{align}
        \label{e:suff-cond-3}
        \left(\abs{\ln{\lambda_{i_s}}}-\left(\ln{\mu_{i_s i_u}}+\ln{\mu_{i_u i_s}}\right)\right)>\left(\left\lceil{\frac{N}{M}}\right\rceil -1\right)(\ell + 1)\ln{\lambda_{\iu}},
        \end{align}
        is satisfied for all \(i=1,2,\ldots,N\), then there exists a contractive cycle \(c\) on \(\G\).
    \end{props}
    \begin{proof}
        The above inequality follows immediately from \eqref{e:suff-cond-2} whenever the number of plants in each partition \(\mathcal{P}_j\), \(j=1,2,\ldots,M\) are either equal (for \(N\%M=0\)) or differ by at most 1 (when \(N\%M\neq 0\)). In both these cases we have \(\abs{\mathcal{P}_j}\leq \lceil{\frac{N}{M}}\rceil\), \(j=1,2,\ldots,M\). Substituting this in \eqref{e:suff-cond-2} results in \eqref{e:suff-cond-3}.
    \end{proof}
    Extreme cases for the above inequality is observed for \(M=1\), when the multiplicative constant in the RHS of \eqref{e:suff-cond-3} is maximum: \(\left(N-1\right)(\ell + 1)\ln{\lambda_{\iu}}\) and for \(M \geq \frac{N}{2}\), when the multiplicative constant in the RHS of \eqref{e:suff-cond-3} is minimum: \((\ell + 1)\ln{\lambda_{\iu}}\). This results in the following immediate corollaries to the above Proposition:
    \begin{coro}
        \label{coro:coro-1}
        If \(M \geq \frac{N}{2}\) and the following inequality 
        \begin{align}
        \label{e:suff-cond-4}
        \left(\abs{\ln{\lambda_{i_s}}}-\left(\ln{\mu_{i_s i_u}}+\ln{\mu_{i_u i_s}}\right)\right)>(\ell + 1)\ln{\lambda_{\iu}},
        \end{align}
        holds for all \(i=1,2,\ldots,N\), then there exists a contractive cycle \(c\) on \(\G\).
    \end{coro}
    \begin{coro}
        \label{coro:coro-2}
        If the following inequality 
        \begin{align}
        \label{e:suff-cond-4}
        \left(\abs{\ln{\lambda_{i_s}}}-\left(\ln{\mu_{i_s i_u}}+\ln{\mu_{i_u i_s}}\right)\right)>(N-1)(\ell + 1)\ln{\lambda_{\iu}},
        \end{align}
        holds for all \(i=1,2,\ldots,N\), then there exists a contractive cycle \(c\) on \(\G\), irrespective of the value of \(M\).
    \end{coro}
\section{Numerical example}
\label{s:numex}
We consider an NCS with \( N = 5 \) discrete-time linear plants and limited communication bandwidth such that \( M = 2 \). Data losses in the two active channels at any time \( t \) are denoted \( \kappa_1(t), \kappa_2(t) \in \mathcal{D}(\ell) \), with \( \ell = 2 \) and \( t \in \mathbb{N} \). Matrices \( A_i \in \mathbb{R}^{2 \times 2} \), \( B_i \in \mathbb{R}^{2 \times 1} \), and \( K_i \in \mathbb{R}^{1 \times 2} \), for \( i = 1, 2, 3, 4, 5 \) are defined as follows. The elements of each \( A_i \) are chosen randomly from the interval \([-2, 2]\), ensuring that each \( A_i \) is Schur unstable. Each \( B_i \) is selected with elements randomly chosen from \([-4, 4]\), and we ensure controllability of each pair \( (A_i, B_i) \). The matrix \( K_i \) is calculated as the discrete-time LQR gain for the pair \( (A_i, B_i) \), using \( Q_i = Q = 5I_{2 \times 2} \) and \( R_i = R = 1 \). The specific matrices are given as:

{\small
\[
A_1 = \begin{pmatrix} 0.0825 & -0.3687 \\ -0.6039 & -0.8382 \end{pmatrix}, \; A_2 = \begin{pmatrix} -0.1379 & -0.0056 \\ 0.6199 & 1.1072 \end{pmatrix},
\]
\[
A_3 = \begin{pmatrix} -0.0719 & -0.0029 \\ 0.6263 & 1.1162 \end{pmatrix}, \; A_4 = \begin{pmatrix} 0.1024 & -0.3600 \\ -0.5975 & -0.8286 \end{pmatrix},
\]
\[
A_5 = \begin{pmatrix} 0.0878 & -0.3470 \\ -0.5946 & -0.8286 \end{pmatrix}.
\]
}
The \( B_i \) matrices are:
{ \small
\[
B_1 = \begin{pmatrix} -0.2825 \\ 3.5699 \end{pmatrix}, \; B_2 = \begin{pmatrix} -0.1856 \\ -1.1214 \end{pmatrix}, \; B_3 = \begin{pmatrix} -0.1350 \\ -1.1353 \end{pmatrix},
\]
\[
B_4 = \begin{pmatrix} -0.1995 \\ 3.5113 \end{pmatrix}, \; B_5 = \begin{pmatrix} -0.2671 \\ 3.5698 \end{pmatrix}.
\]
}
The \( K_i \) matrices are:
{\small
\[
K_1 = \begin{pmatrix} 0.1677 & 0.2231 \end{pmatrix}, \; K_2 = \begin{pmatrix} 0.4442 & 0.8404 \end{pmatrix},
\]
\[
K_3 = \begin{pmatrix} 0.4673 & 0.8519 \end{pmatrix}, \; K_4 = \begin{pmatrix} 0.1691 & 0.2253 \end{pmatrix}, \]
\[K_5 = \begin{pmatrix} 0.1653 & 0.2226 \end{pmatrix}.
\]
}
The scalars \( \lambda_{\is} \), \( \lambda_{\iu} \), \( \mu_{\is\iu} \), and \( \mu_{\iu\is} \) are computed following \citep[Algorithm 2]{atreyee-tcns20} as follows:
{
\small 
\[
\lambda_{1_s} = 0.0043, \; \lambda_{1_u} = 1.0765, \; \mu_{1_s1_u} = 71.2832, \; \mu_{1_u1_s} = 1.9756,
\]
\[
\lambda_{2_s} = 0.0255, \; \lambda_{2_u} = 1.2207, \; \mu_{2_s2_u} = 6.4352, \; \mu_{2_u2_s} = 1.3102,
\]
\[
\lambda_{3_s} = 0.0103, \; \lambda_{3_u} = 1.2434, \; \mu_{3_s3_u} = 6.2035, \; \mu_{3_u3_s} = 1.5210,
\]
\[
\lambda_{4_s} = 0.0077, \; \lambda_{4_u} = 1.0418, \; \mu_{4_s4_u} = 54.8837, \; \mu_{4_u4_s} = 1.7491,
\]
\[
\lambda_{5_s} = 0.0040, \; \lambda_{5_u} = 1.0429, \; \mu_{5_s5_u} = 50.2550, \; \mu_{5_u5_s} = 2.0144.
\]
}
    
    \begin{figure}[htbp]
        \centering
        \includegraphics[height=8cm, width=10cm]{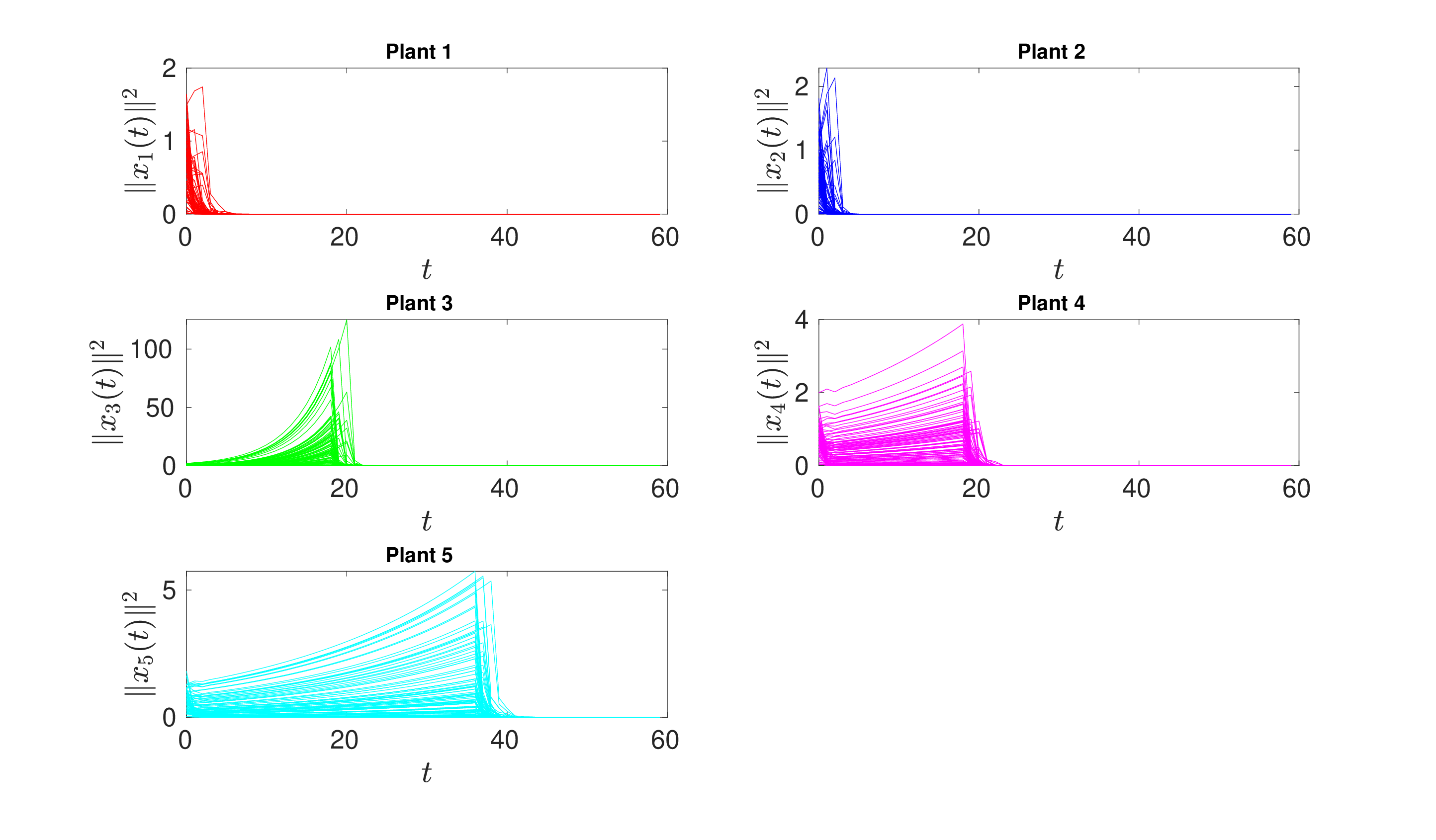}
        \caption{Plot for \(\norm{x_i(t)}^2\) versus \(t\) for all the plants \(i=1,2,3,4,5\), \(t \in [0,60]\)} 
        \label{fig:systems-new}
    \end{figure}
    Here \(\abs{\mathcal{V}}= {N\choose M} = 10\). These vertices are
        \(L_{v_1}=\{1_s,2_u,3_s,4_u,5_u\}\), \(L_{v_2}=\{1_s,2_s,3_u,4_u,5_u\}\),  
        \(L_{v_3}=\{1_s,2_u,3_u,4_s,5_u\}\), \(L_{v_4}=\{1_s,2_u,3_u,4_u,5_s\}\), 
        \(L_{v_5}=\{1_u,2_s,3_s,4_u,5_u\}\), \(L_{v_6}=\{1_u,2_s,3_u,4_s,5_u\}\), 
        \(L_{v_7}=\{1_u,2_s,3_u,4_u,5_s\}\), \(L_{v_8}=\{1_u,2_u,3_s,4_s,5_u\}\), 
        \(L_{v_9}=\{1_u,2_u,3_s,4_u,5_s\}\), \(L_{v_{10}}=\{1_u,2_u,3_u,4_s,5_s\}\).
    Note that the system description above along with the computed scalars satisfy the conditions in Proposition ~\ref{props:suff-cond-2} with \(\mathcal{P}_1=\{1,4,5\}\) and \(\mathcal{P}_2=\{2,3\}\). The \(T\)-factors of the contractive cycles on graphs \(\gdash_1(\Vprime_1, \Edash_1)\) and \(\gdash_2(\Vprime_2, \Edash_2)\) constructed using \(\mathcal{P}_1\) and \(\mathcal{P}_2\) are \(T^1=2\) and \(T^2 = 3\) respectively. Thus we construct the corresponding contractive cycle \(c=v_0,(v_0,v_1),\ldots v_5,(v_5,v_0),v_0\) on \(\G\) following the proof of Lemma \ref{lem:suff-cond-1} with
        \(L_{v_0}=\{1_s,2_s,3_u,4_u,5_u\}\),
        \(L_{v_1}=\{1_u,2_u,3_s,4_s,5_u\}\),
        \(L_{v_2}=\{1_u,2_s,3_u,4_u,5_s\}\),
        \(L_{v_3}=\{1_s,2_u,3_s,4_u,5_u\}\),
        \(L_{v_4}=\{1_u,2_s,3_u,4_s,5_u\}\),
        \(L_{v_5}=\{1_u,2_u,3_s,4_u,5_s\}\).
    and \(T\)-factors \(T_{v_k}=6\), \(k=0,\ldots,5\).
    \begin{figure}[htbp]
        \centering
        \includegraphics[scale=0.17]{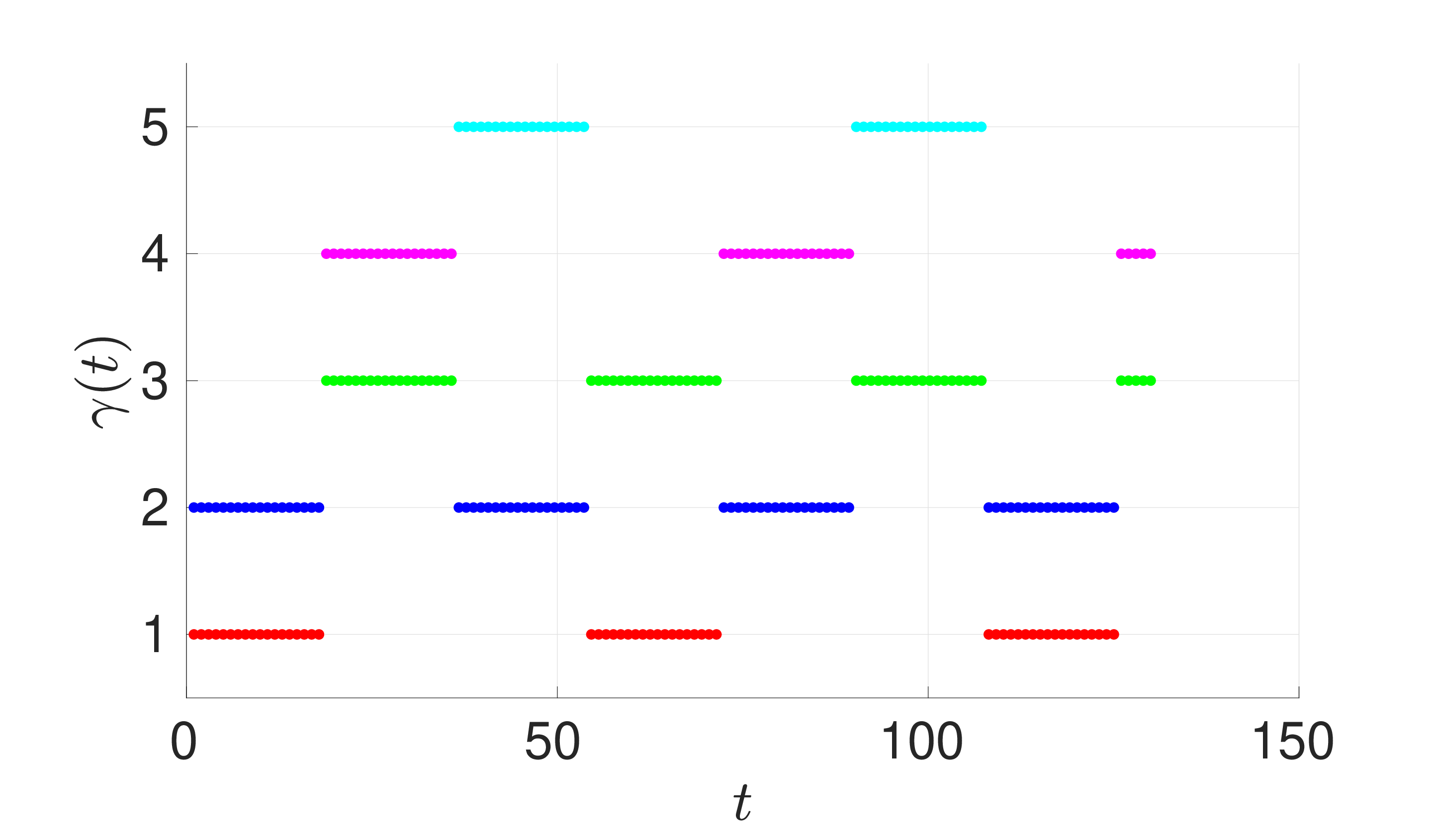}
        \caption{Scheduling logic \(\gamma\) obtained from Algorithm \ref{algo:scheduling_policies-sol-2}}
        \label{fig:gamma-n}
    \end{figure}

    In order to show the \emph{GAS} of the above plants under a scheduling logic \(\gamma\) obtained from Algorithm ~\ref{algo:scheduling_policies-sol-2} with the input cycle as \(c\) and the \(T\)-factors as above, we pick \(100\) different initial conditions at random from \([-10,10]\times [-10,10]\) for \(x_i(0)\) and plot \((\norm{x_i(t)}^2)_{t \geq 0}\) for \(i=1,2,3,4,5\) under admissible data loss signals \((\kappa_1)_{t\geq 0}\) and \((\kappa_2)_{t \geq 0}\) which are chosen at random for each of the \(100\) different choice of initial conditions. The plot obtained from simulations is as in Figure \ref{fig:systems-new} for \(t \in [0,60]\).

    The plot of the stabilizing scheduling logic \(\gamma(t)\) for \(t \in [0,130]\) as obtained by applying Algorithm \ref{algo:scheduling_policies-sol-2} is shown in Figure \ref{fig:gamma-n}.
\section{Acknowledgement}
The author would like to thank Atreyee Kundu for helpful discussions on the problem and comments on a preliminary version of this draft.
\section{Concluding remarks}
\label{s:concln}
    In this article we solved the problem of designing scheduling logic for NCSs with capacity-constrained communication links in the presence of data losses which ensure that all the plants are globally asymptotically stable under all admissible data loss patterns. A next natural direction of work is the co-design of the state-feedback gains \(K_i\), \(i=1,2,\ldots,N\) along with \(\gamma\), investigating similar conditions for continuous-time systems and with probabilistic packet-drops. These are something we are working on and will be reported elsewhere.  
\section{Proofs of our results}
\label{s:proofs}
\subsection{\textbf{Proof of Theorem \ref{thrm:main_res-sol-2}}}
    \begin{proof}
        We consider the NCS under data losses as in Section \ref{s:prob_stat} and its associated directed graph \(\G\). Let \(c=\walk\) be a \emph{contractive cycle} on \(\G\) with associated \(T\)-factors \(T_{v_k}\), \(k=0,1,\ldots,n-1\). 
        We fix an arbitrary plant \(i \in \{1,2,\ldots,N\}\). Using the switched system modelling of the plants, it is sufficient to show that the switching logic \(\sigma_i\) corresponding to \(\gamma\) ensures GAS of plant \(i\). Let \(0=\tau_0 < \tau_1 < \ldots\) be the times at which \(\gamma(t)\) changes its values, as obtained from Algorithm \ref{algo:scheduling_policies-sol-2}. Using the definition of the multiple lyapunov-like functions, we can write for \(t \in \N\) :  \\
        \begin{align}
            \label{e:eqn-1}
            V_{\sigma_i(t)}(x_i(t)) &\leq \lambda_{\sigma_i(\tau_{N_t^{\gamma}})}^{t-\tau_{N_t^{\gamma}}}V_{\sigma_i(\tau_{N_t^{\gamma}})}(x_i(\tau_{N_t^{\gamma}}))
        \end{align}
        Iterating \eqref{e:eqn-1} gives:
        \begin{align}
            \label{e:eqn2}
            V_{\sigma_i(t)}(x_i(t)) &\leq \left[ \prod\limits_{\substack{q=0 \\ \tau_{N_t^{\gamma}+1} \coloneqq t}}^{N_t^{\gamma}}\lambda_{\sigma_i(\tau_q)}^{\tau_{q+1} - \tau_q} \cdot \prod\limits_{q=0}^{N_t^{\gamma}-1}\mu_{\sigma_i(\tau_q)\sigma_i(\tau_{q+1})} \right] \notag \\
            &\qquad V_{\sigma_i(0)}(x_i(0))
        \end{align}
        where \(\lambda_{\is}, \lambda_{\iu}, \mu_{\is\iu}, \mu_{\iu\is}\), \(i=1,2,\ldots,N\) are as stated in Facts \ref{facts:lyapunov-like-1} and \ref{facts:lyapunov-like-2}. \\
        The first term on the RHS of the inequality \eqref{e:eqn2} can be written as:
        \begin{align*}
            \exp{\left(\ln{\left(\prod\limits_{\substack{q=0 \\ \tau_{N_t^{\gamma}+1} \coloneqq t}}^{N_t^{\gamma}}\lambda_{\sigma_i(\tau_q)}^{\tau_{q+1} - \tau_q}\right)} + \ln{\left(\prod\limits_{q=0}^{N_t^{\gamma}-1}\mu_{\sigma_i(\tau_q)\sigma_i(\tau_{q+1})}\right)}\right)}
        \end{align*}
        Now,
        \begin{align}
            \label{e:expamsion}
            &\ln{\left(\prod\limits_{\substack{q=0 \\ \tau_{N_t^{\gamma}+1} \coloneqq t}}^{N_t^{\gamma}}\lambda_{\sigma_i(\tau_q)}^{\tau_{q+1} - \tau_q}\right)} \notag \\ &=\sum\limits_{\substack{q=0 \\ \tau_{N_t^{\gamma}+1} \coloneqq t}}\left(\sum\limits_{p \in \{\is, \iu\}}\mathbbm{1}_{\sigma_i(\tau_q)}(p)(\tau_{q+1} - \tau_q)\ln{\lambda_p}\right)
        \end{align}
        Let \(D_{s}(s,t)\), denote the total number of \(q \in \N\), such that \(\sigma_i(\tau_q) = \is\),  and \(\tau_q \in \left]s:t\right]\). Similarly, let \(D_u(s,t)\) denote the total number of \(q \in \N\), such that \(\sigma_i(\tau_q) = \iu\), and \(\tau_q \in \left]s:t\right]\). Now, since \(\lambda_{\iu} > 1\) and \(0 < \lambda_{\is} < 1\), the RHS of \eqref{e:expamsion} can be written as \(-\abs{\ln{\lambda_{\is}}}D_{s}(0,t) + \abs{\ln{\lambda_{\iu}}}D_u(0,t)\). Now, let \(N_{pr}(s,t)\) denote the number of transitions from mode \(p\) to mode \(q\) in the interval \(\left]s:t\right]\), where \(p,q \in \{\is, \iu\}\), . Then we can right the second term in the RHS of \eqref{e:eqn2} as \(\ln{\left(\prod\limits_{q=0}^{N_t^{\gamma}-1}\mu_{\sigma_i(\tau_q)\sigma_i(\tau_{q+1})}\right)} \notag = \ln{\mu_{\is\iu}}N_{\is\iu}(0,t) + \ln{\mu_{\iu\is}}N_{\iu\is}(0,t)\). Recall that \(\mu_{\is\is} = \mu_{\iu\iu} = 1\). Substituting these in \eqref{e:eqn2} we have
        \begin{align}
            \label{e:compact-form-2}
            V_{\sigma_i(t)}(x_i(t)) \leq \psi_i(t)V_{\sigma_i(0)}(x_i(0)),
        \end{align}
        where
        \begin{align*}
            \N \ni t \mapsto \psi_i(t) &\coloneqq \exp{\left(-\abs{\ln{\lambda_{\is}}}D_{s}(0,t) + \abs{\ln{\lambda_{\iu}}}D_u(0,t) \right.} \\
            & \quad {\left.+ \ln{\mu_{\is\iu}}N_{\is\iu}(0,t)+ \ln{\mu_{\iu\is}}N_{\iu\is}(0,t)\right)},
        \end{align*}
        and \(D_{s}(s,t)\) denotes the total number of \(q \in \N\), such that \(\sigma_i(\tau_q) = \is\),  and \(\tau_q \in \left]s:t\right]\). Similarly, \(D_u(s,t)\) denotes the total number of \(q \in \N\), such that \(\sigma_i(\tau_q) = \iu\), and \(\tau_q \in \left]s:t\right]\). \(N_{pr}(s,t)\) denotes the number of transitions from mode \(p\) to mode \(q\) in the interval \(\left]s:t\right]\), where \(p,q \in \{\is, \iu\}\).
        We know from Facts \ref{facts:lyapunov-like-1} and \ref{facts:lyapunov-like-2} and from the properties of positive definite matrices \citep[Lemma 8.4.3]{Bernstein2009} that \eqref{e:compact-form-2} can be written as \(\norm{x_i(t)} \leq c\psi_i(t)\norm{x_i(0)}\), where \(c=\sqrt{\frac{\text{max}_{p \in \{\is, \iu\}}\lambda_{\text{max}}P_p}{\text{min}_{p \in \{\is, \iu\}}\lambda_{\text{min}}P_p}}\), for \(j \in \{1,2,\ldots,n_i\}\). By Definition \ref{defin:gas} in order to establish GAS of plant \(i\), we need to show that \(c\norm{x_i(0)}\psi_i(t)\) is bounded above by a class \(\mathcal{KL}\) function. We can see that \(c\norm{x_i(0)}\) is a class \(\mathcal{K}_\infty\) function. Thus in order to prove our result, we need to show that \(\psi_i(t)\) is bounded above by a class \(\mathcal{L}\) function. \\
        Using Lemma \ref{lemma:total-time-sol-2}, we define \(\mathcal{T}_m = m\times \left(\sum\limits_{k=0}^{n-1}{T}_{v_k}(\ell + 1)\right)\) and let \(\epsilon = \max\limits_{i=1,2,\ldots,N}\overline{Z}_i(c) < 0\). Then \(\psi_i(t)\) can be written as
        \begin{align}
            \label{e:pf-m-contr-2}
            & \psi_i(t) = \exp \left(\left(-\abs{\ln{\lambda_{\is}}}D_{s}(0,\mathcal{T}_{m-1}) + \abs{\ln{\lambda_{\iu}}}D_u(0,\mathcal{T}_{m-1})\right. \right. \notag \\
            &\quad \left. \left.+ \ln{\mu_{\is\iu}}N_{\is\iu}(0,\mathcal{T}_{m-1}) + \ln{\mu_{\iu\is}}N_{\iu\is}(0,\mathcal{T}_{m-1})\right)\right. \\
            &\left. + \left(-\abs{\ln{\lambda_{\is}}}D_{s}(\mathcal{T}_{m-1},t) + \abs{\ln{\lambda_{\iu}}}D_u(\mathcal{T}_{m-1},t) \right. \right. \notag \\
            &  \left. \left. + \ln{\mu_{\is\iu}}N_{\is\iu}(\mathcal{T}_{m-1},t) + \ln{\mu_{\iu\is}}N_{\iu\is}(\mathcal{T}_{m-1},t) \right)\right). \notag 
        \end{align}
        Notice, from Lemmas \ref{lemma:t-sol-2} and \ref{lemma:total-time-sol-2} it is immediate that:
        {\small \begin{align}
            \label{e:pf-m-ineq-2}
            &\left(-\abs{\ln{\lambda_{\is}}}D_{s}(0,\mathcal{T}_{m-1}) + \abs{\ln{\lambda_{\iu}}}D_u(0,\mathcal{T}_{m-1}) \right. \notag \\
            & \left.+ \ln{\mu_{\is\iu}}N_{\is\iu}(0,\mathcal{T}_{m-1}) + \ln{\mu_{\iu\is}}N_{\iu\is}(0,\mathcal{T}_{m-1})\right)  \notag \\
            &\leq -\abs{\ln{\lambda_{\is}}}(m-1)\times \sum\limits_{k=0}^{n-1}\mathbbm{1}_{v_k}(\ell_{\vbar_k}(i)=\is){T}_{v_k} \notag \\
            &+ \abs{\ln{\lambda_{\iu}}}(m-1)\times \sum\limits_{k=0}^{n-1}\mathbbm{1}_{v_k}(\ell_{v_k}(i)=\iu){T}_{v_k}(\ell + 1) \notag \\
            &+ (m-1)\times \sum\limits_{k=0}^{n-1}\left(\ln{\mu_{\ell_{v_k}(i)\iu}}+\ln{\mu_{\iu\ell_{v_k}(i)}}\right){T}_{v_k} \notag \\
            &= (m-1)\times \overline{Z}_i(c) \leq (m-1)\times \epsilon.
        \end{align}}
        Also, notice that,
        {\small \begin{align}
            \label{e:pf-a-ineq}
            &\left(-\abs{\ln{\lambda_{\is}}}D_{s}(\mathcal{T}_{m-1},t) + \abs{\ln{\lambda_{\iu}}}D_u(\mathcal{T}_{m-1},t) \right. \notag \\
            &\left.+ \ln{\mu_{\is\iu}}N_{\is\iu}(\mathcal{T}_{m-1},t) + \ln{\mu_{\iu\is}}N_{\iu\is}(\mathcal{T}_{m-1},t)\right) \notag \\
            & \leq \ln{\lambda_{\iu}}(t-\mathcal{T}_{m-1}) + \left(\ln{\mu_{\is\iu}}+\ln{\mu_{\iu\is}}\right)\left(\sum\limits_{k=0}^{n-1}{T}_{v_k}\right)(\ell + 1) \notag\\
            &\coloneqq a \text{ (say).}
        \end{align}}
        From \eqref{e:pf-m-ineq-2} and \eqref{e:pf-a-ineq}, it is clear that \(\psi_i(t)\) is upper bounded by \(\exp{\left(-m\abs{\epsilon} + a\right)}\). \\
        Let \(\Gamma=\sum\limits_{k=0}^{n-1}{T}_{v_k}(\ell + 1)\). Let \(\varphi_i \colon [0:t] \rightarrow \R\) be a function connecting \((0, e^{a}+\Gamma)\), \((k\Gamma, e^{(-(k-1)\abs{\epsilon}+a)})\), \((t, e^{(-m\abs{\epsilon}+a)})\), for all \(k=1,2,\ldots,m\). By construction, \(\varphi_i\) is an upper envelope of \(t^{\prime} \mapsto \psi_i(t^{\prime})\), \(t^{\prime} \in [0:t]\). It is continuous, decreasing and tends to \(0\) as \(t \tendsto +\infty\). Hence, \(\varphi_i \in \mathcal{L}\). \\
        Since the plant \(i\), \(i=1,2,\ldots,N\), was selected arbitrarily, and the analysis done for any admissible data loss signal, it follows that the assertion of Theorem \ref{thrm:main_res-sol-2} holds for all plants \(i\) in \eqref{e:plants}.
    \end{proof}
\subsection{\textbf{Proof of Lemma \ref{lem:suff-cond-1}}}
    \begin{proof}
        Let the disjoint partitions of the entire set of \(N\) plants be denoted by \(\Nsets_j\), with the underlying graphs \(\gdash_j(\Vprime_j, \Edash_j)\), \(j=1,2,\ldots,M\). We do the proof in two parts. First we show that since \(M_{\Nsets_j}=1\), \(j=1,\ldots,M\) a contractive cycle \(c_j\), \(j=1,\ldots,M\) can have only a very specific structure. Using the same, we discuss next if \(c_j\) is a contractive cycle, the procedure to obtain the contractive cycle \(c\) for the graph defined for the overall network \(\G\). 
        
        A contractive cycle \(c_j\) on \(\gdash_j(\Vprime_j, \Edash_j)\), \(j=1,\ldots,M\) must satisfy \(c_j=v_1, (v_1,v_2), v_2, \ldots, v_{\abs{\Nsets_j}},(v_{\abs{\Nsets_j}},v_1),v_1\), where \(v_k \in \Vprime_j\), \(k=1,\ldots,\abs{\Nsets_j}\) with labels \(\ell_{v_k}(i)=\iu\) for all \(i=\{1,2,\ldots,\abs{\Nsets_j}\}\setminus \{k\}\) and \(\ell_{v_k}(i)=\is\) for \(i=k\), for all vertices \(v_k\), \(k=1,2,\ldots,\abs{\Nsets_j}\). The reason is that for \(M_{\Nsets_j}=1\) any other choice will cause the cycle to have atleast one vertex which does not get activated at all and thus can never be contractive. If \(c_j\), \(j=1,2,\ldots,M\) are all contractive cycles, then from \eqref{e:contractive-sol-2} there exist positive integers \(T_{v_k}^{j}\), \(k=1,2,\ldots,\abs{\Nsets_j}\) for which \(\overline{Z}_i(c_j)<0\) for each \(i \in \Nsets_j\). 
        
        To obtain the overall contractive cycle \( c \) on \( \mathcal{G} \) along with its associated \( T \)-factors, we proceed as follows: let \( a_j \) be positive integers chosen such that \( a_j \times \sum\limits_{k=1}^{\abs{\mathcal{N}_j}} T^j_{v_k} = K \; \text{(constant)} \) for each \( j=1,2,\ldots,M \). Define \( \widetilde{T}^j_{v_k} = a_j \times T^j_{v_k} \), for \( k=1,2,\ldots,\abs{\mathcal{N}_j} \) for all \( j=1,2,\ldots,M \). Next, we construct a matrix \( \mathcal{M} \) of dimensions \( N \times K \), where each row \( i \), with \( i \in \Nsets_j \) for \( j = 1, 2, \ldots, M \), is filled by \( \widetilde{T}^j_{v_k} \)-many consecutive entries of the form \( \ell_{v_k}(i) \) sequentially for \( k = 1, 2, \ldots, \abs{\Nsets_j} \). Upon completion, each column of \( \mathcal{M} \) contains exactly \( M \) closed-loop (stable mode) labels, due to the partitioning and the property \( M_{\Nsets_j} = 1 \) for all \( j \). Once the matrix \( \mathcal{M} \) is filled, due to the partitioning and the property \( M_{\Nsets_j} = 1 \) for all \( j = 1, 2, \ldots, M \), each column contains exactly \( M \) closed-loop (stable mode) labels. By identifying the unique columns (say there are \( n \) of these), we can define the vertices of a cycle \( c = \cycle \) on \( \G \) by stacking the labels from each unique column into an \( N \times 1 \) vector. This results in vertices \( \vbar_k \in \V \) for \( k = 0, 1, \ldots, n - 1 \). Let \( T_{\vbar_k} \) represent the frequency of each distinct column (corresponding to vertex \( \vbar_k \)) in \( \mathcal{M} \). Thus, the cycle \( c \) on \( \G \), with \( T \)-factors \( T_{\vbar_k} \) for \( k = 0, 1, \ldots, n - 1 \), is contractive due to the constructed properties of \( \mathcal{M} \) and the vertices \( \vbar_k \), satisfying all required conditions for \( i \in \Nsets_j \) for \( j = 1, 2, \ldots, M \).
        
        {\small \begin{align}
            \label{e:suff-cond-1}
            \overline{Z}_i(c) &\coloneqq \left(\sum\limits_{k=0}^{n-1}\left(\mathbbm{1}_{\ell_{\vbar_k}(i)}\left(\ell_{\vbar_k}(i)=\is\right)\wbar_i(\vbar_k) \right. \right. \notag \\ 
            &\quad \quad \qquad \left.+ \mathbbm{1}_{\ell_{\vbar_k}(i)}\left(\ell_{\vbar_k}(i)=\iu\right)\wbar_i(\vbar_k)\times(\ell + 1) \right) \notag \\
            &\quad \qquad + \left.\left(\ln{\mu_{\ell_{\vbar_k}(i)\iu}}+\ln{\mu_{\iu\ell_{\vbar_k}(i)}}\right)\right)T_{\vbar_k} \notag \\
            &= \left(\sum\limits_{k=1}^{\abs{\Nsets_j}}\left(\mathbbm{1}_{\ell_{v_k}(i)}\left(\ell_{v_k}(i)=\is\right)\wbar_i(v_k) \right. \right. \notag \\ 
            &\quad \quad \qquad \left.+ \mathbbm{1}_{\ell_{v_k}(i)}\left(\ell_{v_k}(i)=\iu\right)\wbar_i(v_k)\times(\ell + 1) \right) \notag \\
            &\quad \qquad + \left.\left(\ln{\mu_{\ell_{v_k}(i)\iu}}+\ln{\mu_{\iu\ell_{v_k}(i)}}\right)\right)\widetilde{T}^j_{v_k} \notag \\
            &= a_j\times \overline{Z}_i(c_j) < 0 
        \end{align}}
        
    \end{proof}
\bibliographystyle{ieeetr}
\bibliography{refr}



\end{document}